\documentclass[11pt]{article}
\usepackage{times}
\usepackage{amsmath, amssymb, amsthm}
\usepackage{graphicx}

\advance \topmargin by -\headheight
\advance \topmargin by -\headsep

\evensidemargin \oddsidemargin
\newcommand{\OPT}{{\bf OPT}}
\newcommand{\E}{\mathrm{E}}
\newcommand{\M}{\mathcal{M}}

\newtheorem{theorem}{Theorem}[section]

\newtheorem{lemma}[theorem]{Lemma}
\newtheorem{proposition}[theorem]{Proposition}

\theoremstyle{remark}
\newtheorem{rem}[theorem]{Remark}

\theoremstyle{definition}
\newtheorem{definition}[theorem]{Definition}

\title{Auctions with Online Supply}
\author{Moshe Babaioff\\ Microsoft Research\\ Mountain View, CA 94043\\ moshe@microsoft.com \and Liad Blumrosen\\ Microsoft Research\\ Mountain View, CA 94043\\liadbl@microsoft.com \and Aaron L. Roth\\ Computer Science Department\\ Carnegie Mellon University\\ Pittsburgh, PA 15217\\ alroth@cs.cmu.edu}

\begin{document}
\maketitle
\thispagestyle{empty}

\begin{abstract} We study the problem of selling identical goods to
$n$ unit-demand bidders in a setting in which the total
\emph{supply} of goods is unknown to the mechanism.
Items arrive dynamically, and the seller must make the allocation
and payment decisions online with the goal of maximizing social
welfare. We consider two models of unknown supply: the adversarial
supply model, in which the mechanism must produce a welfare
guarantee for any arbitrary supply, and the stochastic supply model,
in which supply is drawn from a distribution known to the mechanism,
and the mechanism need only provide a welfare guarantee in
expectation.

Our main result is a separation between these two models. We show
that all \emph{truthful} mechanisms, even randomized, achieve a
diminishing fraction of the optimal social welfare (namely, no
better than a $\Omega(\log\log n)$ approximation) in the adversarial
setting. In sharp contrast, in the stochastic model, under a
standard \emph{monotone hazard-rate} condition, we present a truthful
mechanism that achieves a constant approximation. We show that the monotone hazard rate condition is necessary, and also characterize a natural subclass of truthful mechanisms in
our setting, the set of \emph{online-envy-free} mechanisms. All of
the mechanisms we present fall into this class, and we prove almost
optimal lower bounds for such mechanisms.
Since auctions with unknown supply are regularly run in
many online-advertising settings, our main results emphasize the
importance of considering distributional information in the design
of auctions in such environments.

\end{abstract}

\thispagestyle{empty}
\setcounter{page}{0}
\clearpage

\section{Introduction}
Auctions have recently received attention in computer science
because they crystalize many of the incentive issues in algorithmic
game theory, and have direct application to the fast-growing market
for online advertising. This paper belongs to a line of research
that studies \emph{online mechanism design}, which focuses on
markets in which decisions are made dynamically before information regarding the state of the world has been fully revealed.
%Much of the theoretical work on auction design, however, has
%diverged from the realities facing auction designers in several
%important respects.
%Many of the problems facing computational
%mechanisms are inherently online, but much of the previous work has
%focused on the design of truthful mechanisms for \emph{static}
%settings.
Previous work in online mechanism design mainly concerned settings
where \emph{customers} that arrive dynamically compete for buying a
known set of items (see a recent survey \cite{Par07}). However, in
many real-world settings the \emph{supply} arrives dynamically and
the exact number of items for sale is uncertain. This, for example,
is the case in the sale of clicks on banner ads, where the number of
clicks is not known in advance to the seller; Such a seller must
decide which advertisement to show in a fraction of a second after
the item arrives, while the future supply is uncertain.\footnote{
Uncertainty on the supply appears in various environments. More
examples include markets for computing resources and also
traditional markets, like agricultural markets, where produce and
fish continue to arrive after markets has been opened.}
%Liad: the following is a bit out of context - can be deferred.
%When online supply has been studied (e.g.,
%\cite{}), often as a concession to truthfulness, payment has been
%deferred to the end of the auction, which is an unrealistic
%constraint in many environments.
%Liad: not completely accurate.
%Finally, most work on online mechanism design has been either in the
%adversarial setting, when in actuality, mechanism designers have a
%wealth of distributional information at their disposal.

In this work, we investigate a natural online setting, in which a
mechanism must allocate items to a fixed set of bidders when the
supply of items is unknown, and arrives online. We require that the
mechanism allocates items and extracts payment for them as they
arrive. The restriction that the mechanism extract payment at the time of sale is a natural practical constraint, and is satisfied by most real-world markets. Even in markets in which customers are able to defer their payments (such as auctions for search ads), the seller typically calculates payments immediately, which allows customers to better keep track of their spending.
%(Items are \emph{perishable}.
%This models, for example, search and banner ad auctions; in these cases,
%the items for sale are page views, and advertisers must be assigned to
%ad-slots in a fraction of a second after the item arrives).
We introduce a stochastic model where the seller knows how the
supply is distributed, but we do not assume any prior distribution
on the bidders' valuations, nor do we require that the bidders know
how the supply is distributed. One of the
conceptual contributions of our paper is this \emph{hybrid
stochastic model}, in which the supply is drawn from some prior
distribution, but no distributional assumptions are made on the
preferences of the bidders. This captures scenarios such as online advertising,
 in which sellers can easily collect statistics on the supply (e.g.,
number of ad impressions per day) but obtaining statistics on the
actual valuations of the bidders is harder and may requires
modeling, for example, their equilibrium behavior. Most of the
recent work in computer science on online mechanism design has been
in the fully adversarial setting, when in actuality, mechanism
designers have a wealth of distributional information at their
disposal. In economics, at the other extreme, dynamic mechanism
design has been recently studied in a full Bayesian setting that
assumes the existence of prior distributions on the bidders'
preferences.

We wish to maximize social welfare, which is a desirable goal even from the perspective of a for-profit seller that does not have the luxury of operating under monopoly conditions. An economically efficient market (one that maximizes the \emph{combined} welfare of the customers \emph{and} the seller) will be more attractive to customers, and avoids harming the seller in the long term at the expense of short-term profits. In fact, the generalized second price auction currently used to sell search advertisements has social welfare, rather than revenue guarantees \cite{EOS07}.

We explore the cost of ignoring distributional information. We produce a strong separation: Our main
results are \emph{lower bounds} in the adversarial setting, and
\emph{truthful approximation mechanisms} in the stochastic setting.

% We study the problem of auctioning an unknown
%supply of identical items. All bidders are present in the market
%from the beginning, but their willingness to pay for the items is private information. Each bidder $i$ desires a single item, and has a private non-negative value $v_i$
%for receiving an item. Items arrive one at a time, and are
%\emph{perishable} -- they must be allocated as soon as they arrive.
%This models, for example, search and banner ad auctions; In these
%cases, the items for sale are page views, and advertisers must be
%assigned to ad-slots in a fraction of a second after the item
%arrives. We require that our mechanisms be dominant-strategy
%truthful for the bidders ex-post (a stronger requirement than
%truthfulness in expectation), and that they charge winning bidders at the
%moment that they are allocated items.
%%This is
%%a desirable property introduced by \cite{CDF08} which we discuss
%%more in section \ref{RelatedWorkSection}.
%We wish to maximize \emph{social welfare}, the sum of the values of
%the bidders to whom we have allocated items.

Notably, the algorithmic problem that we face is simple. If bidder
valuations were known, then the greedy algorithm which simply
allocated each arriving item to the unsatisfied bidder with the
highest value would achieve optimal social welfare even in the
adversarial supply setting. The difficulty of the problem stems from
the fact that bidders may misrepresent their valuations for personal
gain. Any allocation rule that we design must be associated with a
corresponding payment rule which incentivizes bidders to truthfully
report their valuations.
%Our truthfulness notion requires that
%truthful reporting will be a dominant strategy for every realization
%of supply and randomness of the mechanism.
As we shall show, the incentive constraint proves to be an insurmountable barrier to
developing mechanisms guaranteeing a constant approximation to
social welfare in the adversarial supply setting, but can be
overcome in the stochastic supply setting.
% under
%the standard \emph{monotone hazard rate} condition.

\subsection{Our Results}

%In section \ref{ModelSection} we define our model and notation. In
%section \ref{AdversarialSection}
We first consider the adversarial supply setting in which welfare
guarantees are required to hold for any realization of supply. Our
first main result are lower bounds on the
%social welfare
approximation obtainable by truthful mechanisms:

\vspace{2mm} {\bf Theorem:} \emph{Every truthful mechanism achieves
a diminishing fraction (in the number of bidders) of the optimal
social welfare. Specifically, no deterministic truthful mechanism
achieves better than $n$-approximation and no randomized truthful
mechanism achieves better than $\Omega(\log \log n)$-approximation.}
 \vspace{2mm}

The linear lower bound is simple, and is in the spirit of the lower
bound given by Lavi and Nisan \cite{LN05} for a model in which bidders
that arrive online bid for a fixed set of expiring items. %, showing that no
%deterministic mechanism can achieve better than an $n$-approximation
%to social welfare.
We note that an $n$-approximation to social welfare can be achieved
by the trivial mechanism which simply allocates the first item to
the highest bidder at the second highest price, and does not
allocate any additional items. The randomized lower bound is more
technically challenging. To prove it we give a characterization of truthful mechanisms in our setting, and a distribution over bidder values. From this, we derive a system of equations that can be simultaneously satisfied only if there exists a mechanism which achieves a strong welfare guarantee when given this distribution over bidders. We show that no such satisfying assignment exists, which gives the lower bound.
%, showing that
%no randomized mechanism can achieve better than an $\Omega(\log\log
%n)$ approximation to social welfare.

If we further require that our mechanisms be \emph{online envy-free}
(a desirable fairness property that we define in section
\ref{section:envy-free}), we can strengthen the above lower bound to
show that no randomized truthful mechanism can achieve better than
an $\Omega(\log n/\log\log n)$ approximation to social welfare. We
show that this last result is almost tight by giving a truthful,
online-envy-free mechanism which achieves a $\log n$ approximation
to social welfare. We leave open the problem of closing the gap
between our upper and lower bounds for non-envy-free randomized
mechanisms, which seems to require different techniques. All our
lower bounds hold even for algorithms that are not computationally
restricted, while our upper bounds follow from computationally efficient mechanisms.

 Given the impossibility in the
adversarial model, we then
%In section \ref{StochasticSection} we consider
consider the stochastic supply setting in which supply is drawn from
a distribution $D$ known to the mechanism, and welfare guarantees
are required to hold in expectation over $D$. We make the assumption
(standard in mechanism design in other contexts) that $D$ has a non-decreasing hazard
rate\footnote{A cumulative distribution $F$ with density $f$ has
\emph{non-decreasing hazard rate} (sometimes called \emph{monotone hazard rate}) if $\frac{f(x)}{1-F(x)}$ is non-decreasing with
$x$.  }. Our second main result is a positive one:

\vspace{2mm} {\bf Theorem:} \emph{There exists a truthful mechanism
that achieves a constant
%$16\frac{7}{8}$
approximation to social welfare when supply is drawn from a known distribution with non-decreasing hazard rate.}
%\emph{For every distribution on the supply that exhibits non-decreasing hazard rate, there exists a truthful mechanism that obtains a constant fraction of the optimal social welfare.}
 \vspace{2mm}

%The mechanism we construct are prompt and ex-post truthful.

This mechanism is simple, deterministic, computationally efficient,
and easy to implement, but it's analysis is surprisingly subtle. We
stress that the incentive properties of the mechanisms we give do
not rely on any distributional information. In particular, truthful
bidding is a dominant strategy for every set of bids, for every
supply, and for any realization of the coin flips of the mechanism
(truthful "in the universal sense", see \cite{NR01,DNS06}), not only
in expectation. Truthfulness in expectation over supply realization
would require that all the bidders and the seller share the same
beliefs on how the supply is distributed. This is unlikely either
because bidders do not have the resources needed for
estimating these priors, or, because they may have
private information that creates heterogeneity in their beliefs
(see, e.g., \cite{AL01}).\footnote{We note that in the stochastic
setting, we can achieve optimal welfare using expected VCG prices if
we were to require only truthfulness \emph{in expectation} over the
supply $\ell$. However, this seems to be a weak solution concept,
since bidders may be motivated to misrepresent their valuations if
their understanding of the supply distribution $D$ differs from the
mechanism's, or if they are not risk-neutral. In this paper, we show
that positive results can be achieved even with this stronger
solution concept.}

We also show that the non-decreasing hazard rate assumption is
necessary: no deterministic mechanism can achieve a constant
approximation (or, in particular, better than an $\Omega(\sqrt{\log
n/\log\log n})$ approximation) to social welfare over arbitrary
distributions. As mentioned, our mechanism is deterministic, and
does not involve randomization techniques used in previous papers
for obtaining truthful approximations (like random sampling, see
\cite{HK07,DNS06}).

%Also, all decisions are made by the seller in an online fashion, both on allocation and payments.
%In addition, our truthful mechanism is deterministic and does not require random sampling techniques that were used in previous work (e.g., \cite{jason01,DNS06,jasonsurvey}).

%Finally, in section \ref{ExtensionSection} we present some
%extensions to our model.
Finally, we also consider the setting in which the bidders preferences may exhibit
complementarities for multiple items (\emph{increasing} marginal
utilities).
%We show that all of our mechanisms also
%apply to the case in which bidders desire multiple items, and have a
%valuation function exhibiting \emph{non-increasing marginal
%utility}.
%We also consider the case of bidders with increasing
%marginal utility for multiple items,
We study the the extreme case of \emph{knapsack valuations} (or
single-minded bidders) and show strong lower bounds (even in the
stochastic supply setting) on the competitive ratio that any
algorithm can achieve, even without incentive constraints. We
provide an algorithm with an exactly matching competitive ratio to
prove that our lower bound is tight.
% Lastly, we extend our setting to multiple item types, and show how to obtain an $O(d)$ approximation to social welfare when supply from each item is drawn from a distribution satisfying the monotone hazard rate condition, and when each bidder is interested in one item out of at most $d$ types of items.

\subsection{Related Work}
\label{RelatedWorkSection}

The works most related to ours are Mahdian and Saberi \cite{MS06},
Cole, Dobzinski, and Fleischer \cite{CDF08} and Lavi and Nisan
\cite{LN05}. Mahdian and Saberi  \cite{MS06} is the only other work that we are
aware of to study mechanisms in which the supply is unknown and
arrives online. They study the sale of multiple types of goods to
bidders who desire only a single item, and wish to design mechanisms
to maximize revenue. They consider only the adversarial supply
setting, and allow extracting all payments when the entire supply
has been exhausted
% (i.e. their mechanisms are not prompt)
. In this model, they give a truthful mechanism that is constant
competitive with respect to the optimal auction that is restricted
to selling all items at a single price, and show a lower bound of
$(e+1)/e$. Their mechanism is randomized, and is based on
random-sampling techniques to achieve truthfulness.

Cole, Dobzinski, and Fleischer \cite{CDF08} introduce the concept of
\emph{prompt mechanisms}, which impose the natural condition that
bidders learn their payment immediately upon winning an item. They
observe that mechanisms which are not prompt are often unusable,
because, e.g.,  they tie up bidders to the auction for too long,
they make debt collection difficult, and they require a high level
of trust in the auctioneer.
%(See \cite{CDF08} for more
%discussion of these issues)
They study prompt mechanisms for a problem in which the supply of
$m$ expiring items is fixed and known to the mechanism, but the bidders
arrive and depart online. They wish to maximize social welfare, and
give a truthful $\log m$ competitive mechanism, and show a lower
bound of $2$ even for randomized mechanisms. Similar models of
online auctions with expiring goods were studied earlier by Lavi and
Nisan \cite{LN05} and by Hajiaghayi et al. \cite{HKMP05}. These models
relate to ours since the allocation decisions for items with
expiration date (airline tickets, for instance) must be made online. In these papers, however, there is no uncertainty on the
supply and bidders arrive and depart over time. More on online
auctions, which were first discussed by Lavi and Nisan \cite{LN04},
can be found in the survey \cite{Par07}.

%Online auctions were first discussed by Lavi and Nisan \cite{LN04},
%followed by a series of paper on online mechanism design (see the
%survey \cite{Par07}). In particular, the work on online auctions
%with expiring goods (\cite{LN05,HKMP05}) is related to our model as
%the allocation decision for items with expiration date (airline
%tickets, for instance) should be made on the fly. In these papers,
%however, there is no uncertainty on the supply and bidders arrive
%and depart over time.
A recent line of papers studies online mechanism design in a
Bayesian setting (\cite{CPS05,AS07,BV06}), where welfare-maximizing,
and even budget balanced, generalizations of VCG mechanisms are
presented for online settings. Our paper does not assume a Bayesian
preference model and, as our lower bounds show, socially-efficient
outcomes cannot be truthfully implemented. In the economics
literature, stochastic supply has not been studied in  many papers.
Most of this work (see, for example, \cite{Jei99,NP07}) studied a
Bayesian model, and focused on the characterization of equilibrium
prices. Uncertain supply models can be viewed as more complicated
versions of the classic sequential auctions model, which is technically
hard to analyze even without uncertainty on the supply (see, e.g.,
\cite{MW82,MV93}).

While our paper focuses on auctions for identical goods with
%decreasing marginal utilities, we briefly discuss other domains in Section \ref{ExtensionSection}.
bidders that are interested in a single item, we briefly discuss a more general domain in which
single minded bidders are interested in multiple items in Section \ref{sec:knapsack}.
Knapsack auctions (or auctions for
single-minded bidders) were studied by \cite{AH06,DN07} for static
settings with known supply.
%Aaron: Doesn't seem necessary in related work section
%We also touch auctions for heterogenous
%items. Interestingly, for the well studied class of submodular
%valuations in combinatorial auctions, one could achieve good
%approximation results in the online-supply setting (without taking
%incentives into account). That is, the simple greedy 2-approximation
%algorithm in \cite{LLN06} allocates the items one by one in
%arbitrary order and thus can be used when the supply arrives online
%and with uncertain quantities. An interesting open question is
%whether positive results similar to ours can be truthfully attained
%in the presence of reasonable priors on the supply also for
%heterogeneous items (our lower bounds for the adversarial settings
%still hold in this case, of course).

\vspace{2mm}

We proceed as follows. After presenting our formal model in the next
section, we present our main results in Sections
\ref{AdversarialSection} (adversarial supply) and
\ref{StochasticSection} (stochastic supply). We then discuss online-envy-free mechanisms in Section
\ref{section:envy-free} and strengthen our lower bounds, and consider Knapsack valuations in Section
\ref{sec:knapsack}.

\section{Model and Definitions}
\label{ModelSection}
We consider a set of $n$ bidders $\{1,\ldots,n\}$, each desires a single item from a set of identical items (except in Section~\ref{sec:knapsack} in
which we expand our model to agents interested in multiple items.)
Each bidder has a non-negative valuation $v_i$ for an item.
A \emph{mechanism} $\M$ is a (possibly randomized) allocation rule
paired with a payment rule. Bidders report their valuations to the
mechanism before any item arrives, and the mechanism assigns items as they arrive to
bidders, and simultaneously charges each bidder $i$ some price $p_i$.
When $\ell$ items arrive
and bidders have submitted bids $v_1',\ldots,v_n'$, we denote the
outcome of the mechanism by $\M_\ell((v_1',\ldots,v_n'),r)$ where
$r$ is a random bitstring which may be used by randomized
mechanisms. We note that the mechanism is unaware of $\ell$, as it only encounters the items one at a time as they arrive. We will leave out the $r$ when it is clear from
context. We adopt standard notation and write $v'_{-i}$ to denote
the set of valuations reported by all bidders other than bidder $i$.
A bidder $i$ who receives an item obtains utility
$u_i(v_i;\M_\ell(v'_1,\ldots,v'_n)) = v_i - p_i$. Bidders who do not
receive an item obtain utility 0. Bidders wish to maximize their own
utility, and may misrepresent their valuations to the mechanism in
order to do so.

We require that our mechanisms be \emph{truthful}: that bidders
should be incentivized to report their true valuations, regardless
of the bids of others or the realizations of the supply. Following the literature (e.g. Goldberg et al. \cite{GH03}, Guruswami et al. \cite{GH05}) we define a randomized truthful mechanism to be a probability distribution over deterministic truthful mechanisms.
\begin{definition}
A mechanism $\M$ is \emph{(ex-post) truthful} if for every bidder
$i$ with value $v_i$, for every set of bids $v'_{-i}$, for every
alternative bid $v'_i$ and for every $r$ and $\ell$:
$u_i(v_i;\M_\ell((v_i,v'_{-i}),r)) \geq
u_i(v_i;\M_\ell((v'_i,v'_{-i}),r))$
\end{definition}
%We will be concerned only with ex-post truthful mechanisms in this
%paper, and so will refer to them simply as \emph{truthful
%mechanisms}.
We will assume that bidders submit their true valuations to truthful
mechanisms, since it is a dominant strategy for them to do
so.

 Without loss of
generality, we imagine that $v_1,\ldots,v_n$ are written in
non-increasing order. The social welfare achieved by a mechanism is
the sum of the values of the bidders to whom it has allocated items,
which we denote by $W(\M_\ell((v_1,\ldots,v_n),r))$. When $\ell$
items arrive, we will denote the optimal social welfare by
$\OPT_\ell = \sum_{i=1}^\ell v_i$.
% Moshe: we never defined OPT
When $\ell$ is drawn from a distribution $D$ over the support
(w.l.o.g.) $\{1,...,n\}$, we define $\OPT = \E_{\ell}[\OPT_\ell] =
\sum_{i=1}^n \OPT_i\cdot \Pr[l=i] $.

We will be concerned with approximation guarantees to social welfare
in both the \emph{adversarial supply} setting and the
\emph{stochastic supply} setting.

\begin{definition}
A mechanism $\M$ achieves an $\alpha$-approximation to social welfare in the \emph{adversarial supply} setting if for \emph{every} supply $\ell$:
$\frac{\OPT_\ell}{\E_r[W(\M_\ell((v_1,\ldots,v_n),r))]} \leq \alpha$

When $\ell$ is drawn from a distribution $D$, a mechanism $\M$ achieves an $\alpha$-approximation to social welfare in the \emph{stochastic supply} setting if:
$\frac{\E_{\ell}[\OPT_\ell]}{\E_{\ell,r}[W(\M_\ell((v_1,\ldots,v_n),r))]} \leq \alpha$
\end{definition}

In the stochastic setting, we will assume unless otherwise specified that $D$ satisfies the \emph{non-decreasing hazard rate} condition:
\begin{definition}
The hazard rate of a distribution $D$ at $i$ is:
$h_i(D) = \frac{\Pr[\ell = i]}{\Pr[\ell \geq i]}$. We write simply
$h_i$ when the distribution is clear from context.

$D$ satisfies the \emph{non-decreasing hazard rate} condition if $h_i(D)$ is a non-decreasing sequence in $i$.
\end{definition}
The non-decreasing hazard rate condition is standard in mechanism
design (see, for example, \cite{M81,Krishna02} and recent
computer-science work \cite{CHK07,HR08}), and is satisfied by many
natural distributions, including the exponential, uniform, and
binomial distributions.

One might also consider an intermediate model in which supply is drawn from a distribution satisfying the non-decreasing hazard rate condition, but the distribution is unknown to the mechanism. However, we note that since point distributions satisfy the hazard rate condition, adversarial supply is a special case of this model, and so our lower bounds apply.

\section{Adversarial Supply}
\label{AdversarialSection}
In this section we consider the adversarial model in which we do not have a distribution over supply and we require a good approximation to social welfare for any number of items that arrive.
We first show that deterministic truthful mechanisms cannot achieve any approximation better than the trivial $n$-approximation. We then consider randomized mechanisms, and give a lower bound of $\Omega(\log\log n)$, proving in particular that no constant approximation is possible.

\subsection{Deterministic Mechanisms}
We begin by proving that deterministic mechanisms can only achieve a trivial approximation.
We present a sketch of the proof and defer the details to Appendix~\ref{app:det-adversarial}. First, we characterize deterministic truthful mechanisms by two useful observations:
\begin{lemma}
\label{priceLemma}
For every truthful mechanism and for any realization of items, the price $p_b$ that bidder $b$ is charged upon winning (any) item is independent of his bid.
\end{lemma}
\begin{lemma}
\label{orderLemma}
For every truthful mechanism and for any realization of items, if bidder $b$ wins an item, which item bidder $b$ wins is independent of his bid whenever $p_b < v_b$.
\end{lemma}
\begin{theorem}
\label{thm:det-advers-lb}
No deterministic truthful mechanism can achieve better than an $n$ approximation to social welfare.
\end{theorem}
\begin{proof}
(Sketch) We show that if the mechanism achieves any finite approximation to social welfare, every bidder has a bid such
that he is allocated the first item. Applying lemmas \ref{priceLemma} and \ref{orderLemma}, we conclude that any
deterministic truthful mechanism that achieves a finite approximation to social welfare can only sell a single item, which implies that it cannot achieve better than an $n$ approximation when all bidders have the same value for an item. See the appendix for further details.
\end{proof}

\subsection{Randomized Mechanisms}

\subsubsection{An $\Omega(\log\log n)$ lower bound}
We next present our first main result, a lower bound
for randomized truthful mechanisms.
\begin{theorem}
\label{generalLowerBound}
No truthful randomized mechanism can achieve an $o(\log\log n)$ approximation to social welfare when faced with adversarial supply.
\end{theorem}

\begin{proof}

A truthful randomized mechanism is simply a probability distribution over deterministic  truthful mechanisms. To prove our randomized lower bound, we will exhibit a distribution over bidder values such that no deterministic truthful mechanism achieves a good approximation to welfare in expectation over this random instance. By Yao's min-max principle, this is sufficient to prove a lower bound on randomized mechanisms.

We define a distribution $V$ with support over values $1/2^i$ for $0 \leq i \leq \log n-1$. For each realization $v \in V$, we let: $\Pr[v = 1/2^i] = 2^i/(n-1)$. Therefore, we have $\Pr[v \geq 1/2^i] = (2^{i+1}-1)/(n-1)$ and $E[v | v \geq 1/2^i] = (i+1)/(2^{i+1}-1)$.

\begin{lemma}
\label{OPTLemma}
Consider a set of $n$ valuations drawn from $V$ and let $\OPT_k$ denote the sum of the $k$ highest valuations from the set. Then:
$\E[\OPT_k] \geq H_{k+1}-1$
where $H_{k+1}$ denotes the $k+1$st harmonic number. In particular, $\E[\OPT_k] > (\log k)/2$.
\end{lemma}
\begin{proof}
We defer this proof to Appendix~\ref{app:adver-proofs}.
\end{proof}

By Lemma~\ref{priceLemma} and Lemma~\ref{orderLemma}, we may characterize deterministic truthful mechanisms as follows: The mechanism assigns to each bidder $b$ a bin $i_b$ and a threshold $t_b$. $i_b$ and $t_b$ are independent of $b$'s bid $v_b$, but are assigned such that at most one bidder in each bin can have a bid above his threshold.\footnote{An example of such a function is for each bidder's threshold to be the highest bid of any other bidder in his bin. This results in exactly one bidder (the highest) having a bid above his threshold, while maintaining the property that each bidders threshold is independent of his bid.} If $v_b > t_b$, $b$ wins item $i$ (if it arrives) at price $t_b$. Equivalently, we may imagine the mechanism operating by ordering bidders in some permutation $\pi$ such that for all $i$, every bidder in bucket $i$ is ordered before every bidder in bucket $j > i$. When the first item arrives, the mechanism offers it to each bidder at their threshold price, in order of $\pi$ until some bidder $b$ accepts. We continue in this manner, offering the next item to bidders starting at $b+1$ until one accepts, etc.

We construct a distribution over instances by drawing each bidder's valuation independently from the distribution $V$ described above. Since bidder's thresholds and buckets are independent of their own bids, each value encountered by the mechanism when making offers in order of $\pi$ is distributed randomly according to $V$ (note that although the values are distributed randomly, they need not be independent of each other). We may assume without loss of generality that each threshold $t_b = 1/2^{c_b}$ for some $c_b \in {0,\ldots,\log n-1}$.

When all $n$ items arrive, the expected welfare achieved by a mechanism is:
$\sum_{b=1}^n \Pr[v_b \geq \frac{1}{2^{c_b}}]\cdot E[v_b | v_b \geq \frac{1}{2^{c_b}}] = \frac{1}{n-1}\sum_{b=1}^n(c_b + 1)$.
Let $N_{b}$ denote the number of items sold by a mechanism after making offers to $b$ bidders. Then we have more generally, when $k$ items arrive, the expected welfare achieved by a mechanism is:
$\sum_{b=1}^n \Pr[v_b \geq \frac{1}{2^{c_b}}]\cdot E[v_b | v_b \geq \frac{1}{2^{c_b}}]\cdot \Pr[N_{b-1} < k] = \frac{1}{n-1}\sum_{b=1}^n(c_b + 1)\Pr[N_{b-1} < k]$.
If our mechanism achieves an $\alpha$ approximation to social welfare, we therefore have the following $n$ constraints on the values of $c_b$ chosen by the mechanism. For all $1 \leq k \leq n$:
\begin{equation}
\label{lowerBoundConstraint}
\sum_{b=1}^n(c_b + 1)\Pr[N_{b-1} < k] \geq \frac{(n-1)\OPT_k}{\alpha} \geq \frac{(n-1)\log k}{2\alpha}
\end{equation}
where the last inequality follows from Lemma~\ref{OPTLemma}. After offering the item to $b$ bidders, the expected number of sales is $E[N_b] = 1/(n-1)\cdot \sum_{i=1}^b(2^{c_b+1}-1)$.

By a Chernoff bound: $\Pr[N_{b-1} < k] \leq \exp(-(\frac{E[N_{b-1}]}{2}-k+1)) \leq \exp(-\frac{\sum_{i=1}^{b-1}2^{c_i}}{n-1}+k)$.
Let $b_{k}$ be the first index such that $\sum_{i=1}^{b_k}2^{c_i} \geq (n-1)\cdot k$. Then by plugging our bound into constraint \ref{lowerBoundConstraint}, we have for all $k$:
$$\sum_{i=1}^{b_k}(c_i + 1) + \sum_{i=b_k+1}^n\frac{(c_i+1)}{\exp(\frac{\sum_{j=b_k+1}^{i-1}2^{c_j}}{n-1})} \geq \frac{(n-1)\log k}{2\alpha}$$
\begin{lemma}
\label{notTooBigLemma}
For $c_i \in [0,\log n-1]$:
$$\sum_{i=b_k+1}^n\frac{(c_i+1)}{\exp(\frac{\sum_{j=b_k+1}^{i-1}2^{c_j}}{n-1})} < 2.5\cdot n$$
\end{lemma}
\begin{proof}
We defer the proof of this technical lemma to Appendix~\ref{app:adver-proofs}.
\end{proof}

So, for all $k$, there must exist an integer $b_k$ such that simultaneously the two equations hold:
$\sum_{i=1}^{b_k}c_i \geq \frac{(n-1)\log k}{2\alpha} - (2.5\cdot n + b_k)$,  and $\sum_{i=1}^{b_k-1}2^{c_i} < (n-1)\cdot k$. In particular, if  $k \geq 2^{15\alpha}$ and $n \geq 30$, then $\frac{n\log k}{4\alpha} \leq \frac{(n-1)\log k}{2\alpha}-3.5n$. Therefore, there must exist integers $b_k$ to satisfy the equations:
\begin{equation}
\label{firstEqn}
\sum_{i=1}^{b_k}c_i \geq \frac{n\log k}{4\alpha}
\end{equation}
\begin{equation}
\label{secondEqn}
\sum_{i=1}^{b_k-1}2^{c_i} < n\cdot k
\end{equation}

We will consider the smallest such set of $b_k$: For all $k$, we will have that $\sum_{i=1}^{b_k}c_i \geq \frac{n\log k}{4\alpha}$, but $\sum_{i=1}^{b_k-1}c_i < \frac{n\log k}{4\alpha}$. Note that if we reduce a larger $b_k$ in this manner, inequality \ref{secondEqn} continues to hold, and so this is without loss of generality.

We let $k = 2^{15\alpha}$ and consider the sequence of integers $k, 2k, 4k, \ldots, 2^tk$ such that $n \geq 2^tk > n/2$. For $j \geq 1$ we write $\Delta_k^j = (b_{2^jk} - b_{2^{j-1}k})$, and $\Delta_k^0 = b_k$. We note that from inequality \ref{firstEqn} and our assumption on the $b_k$, we have: $\sum_{i=b_{2^{j-1}k}}^{b_{2^jk}}c_i \geq \frac{n(\log k + j)}{4\alpha} - \sum_{i=1}^{b_{(k\cdot2^{j-1})-1}}c_i \geq \frac{n}{4\alpha}$.

Exponentiating both sides and applying the AM-GM inequality we have:
\begin{eqnarray*}
2^{n/(4\alpha \Delta_k^j)} \;\;\; \leq \;\;\;  \left(\prod_{i=b_{2^{j-1}k}}^{b_{2^jk}}2^{c_i}\right)^{1/\Delta_k^j}
\;\;\; \leq \;\;\; \frac{\sum_{i=b_{2^{j-1}k}}^{b_{2^jk}}2^{c_i}}{\Delta_k^j}
\;\;\; \leq \;\;\; \frac{n(2^jk+1)}{\Delta_k^j}
\end{eqnarray*}
where the last inequality follows from inequality \ref{secondEqn}.
This gives us:
$\Delta_k^j \geq \frac{n}{4\alpha(\log n + \log(2^{j+1}k) - \log \Delta_k^j)}$.
We can expand the above recursive bound to isolate $\Delta_k^j$ and find $\Delta_k^j = \Omega(n/(\alpha(j+\alpha)))$.
%\begin{eqnarray*}
%\Delta_k^j &\geq& \left(\frac{n}{4\alpha}\right)\frac{1}{\log(2^{j+1}k) + \log(\alpha) + \log(\log(2^{j+1}k) + \log(\alpha) + \log(\log(2^{j+1}k) + \log\alpha + \ldots  )\ldots )} \\
%&\geq& \left(\frac{n}{4\alpha}\right)\frac{1}{(\log(2^{j+1}k) + \log\log(2^{j+1}k) + \log\log\log(2^{j+1}k) + \ldots) + (\log \alpha + \log\log \alpha + \log\log\log\alpha+ \ldots)} \\
%&\geq& \left(\frac{n}{8\alpha}\right)\frac{1}{\log(2^{j+1}k) + \log(\alpha)} \\
%&\geq& \left(\frac{n}{8\alpha}\right)\frac{1}{j + 16\alpha + \log(\alpha)} \\
%&\geq&\frac{n}{136\alpha(j+\alpha)}
%\end{eqnarray*}

We recall that $n > b_{2^t k} = \sum_{i=0}^t\Delta_k^i$. Using the above bound, we see that $n$ is at least $\sum_{i=0}^t \Omega(n/(\alpha(i+\alpha))) = \Omega(\frac{n\log(t/\alpha)}{\alpha})$.
%\begin{eqnarray*}
%n &\geq& b_{2^t k} \\
%&=& \sum_{i=0}^t\Delta_k^i \\
%&\geq& \sum_{i=0}^t\frac{n}{136\alpha(i+\alpha)} \\
%&=&\frac{n}{136\alpha}(H(t+\alpha)-H(\alpha)) \\
%&=& \Theta\left(\frac{n\log(t/\alpha + 1)}{\alpha}\right) \\
%&\geq& \Theta\left(\frac{n\log(t/\alpha)}{\alpha}\right)
%\end{eqnarray*}
Therefore, we have $\alpha \geq \Theta(\log(t/\alpha))$ and so $\alpha \geq \Theta(\log t)$.
 %and so there is a constant $d$ such that $\alpha + d\log \alpha \geq d\log t$. Therefore:
%$$\alpha \geq \Theta(\log t)$$
We recall that $k = 2^{15\alpha}$ and $2^t k = 2^{15\alpha + t} \leq n$. $t$ is therefore constrained such that: $\log n \geq 15\alpha + t \geq \Theta(t)$.
And so we may take $t$ to be as large as $\Theta(\log n)$, giving us a lower bound of
$\alpha \geq \Theta(\log\log n)$.

\end{proof}

\subsubsection{A truthful $\log n$-approximation mechanism}

%We have proven strong lower bounds on the welfare-approximation
%achievable by ex-post truthful mechanisms in the adversarial
%setting.
Here we show a simple randomized %online-envy-free
mechanism that achieves a $\log n$ approximation to social welfare.
In Section \ref{section:envy-free} we show that this is nearly
optimal for the natural class of "online envy-free" mechanisms.

Let \textbf{RandomGuess} be the mechanism that selects a supply $g
\in \{2, 4, 8, \ldots, 2^i, \ldots, n\}$ uniformly at random, and
considers only the highest $g$ bidders  according to permutation
order.  When an item arrives the mechanism sells it to the first of
the remaining such bidders and charges him $v_{g+1}$.\footnote{The
authors thank Andrew Goldberg for suggesting this mechanism, which
is a significant simplification of our original mechanism.}

% sells the first $g$ items that arrive to the top $g$ bidders in some bid-independent order, and charges them the $g+1$st highest bid $v_{g+1}$. \footnote{The authors thank Andrew Goldberg for suggesting this mechanism, which is a significant simplification of our original mechanism.}

%\newline\newline
%\textbf{RandomGuess}\footnote{The authors thank Andrew Goldberg for suggesting this mechanism, which is a significant simplification of our original mechanism}:
%\begin{enumerate}
%  \item Fix a uniformly random permutation $\pi$ on the bidders.
%  \item Solicit bids, and denote them $v_1,\ldots,v_n$ in non-increasing order.
%  \item Choose $g \in \{2, 4, 8, \ldots, 2^i, \ldots, n\}$ uniformly at random
%  \item Sell each of the first $g$ items that come in to the highest $g$ bidders in the order in which they appear in $\pi$, and charge them $v_{g+1}$.
%    If more than $g$ items arrive, do not sell any additional items.
%\end{enumerate}

%\begin{theorem}
%RandomGuess is truthful and envy-free.
%\end{theorem}

\begin{proposition}
\label{RandomGuessFactor} RandomGuess is truthful % , online-envy-free, and
and achieves a $\log n$ approximation to social welfare.
\end{proposition}
\begin{proof}
We defer this proof to Section \ref{app:adver-proofs} in the Appendix.
\end{proof}

We leave open the problem of closing the gap between the $\log n$
factor achieved by RandomGuess and the $\Omega(\log\log n)$ lower
bound of Theorem \ref{generalLowerBound}. In section \ref{section:envy-free} we strengthen this lower bound to $\Omega(\log n/\log\log n)$ for the class of \emph{online-envy-free} mechanisms, also defined in section \ref{section:envy-free}.
% for non-envy-free mechanisms.
We conjecture that RandomGuess is optimal.

\section{Stochastic Supply}
\label{StochasticSection}
Given the strong lower bounds we have shown in the adversarial setting, we now consider the stochastic setting in which supply is drawn from some distribution $D$ known to the mechanism.
In this section, we give our second main result, a deterministic %online-envy-free
truthful mechanism that achieves an
$O(1)$-approximation to social welfare for any distribution with
non-decreasing hazard rate. At the end of this section we show that the monotone hazard rate condition is actually necessary to achieve constant approximation.

%To present the algorithm we first fix some notation.
%Denote by $\OPT_i$ the optimal social welfare achieved when $\ell = i$.
%That is, $\OPT_i = \sum_{j=1}^i v_j$ when values are sorted in non-increasing order.
%Additionally, let $\OPT$ denote the optimal expected welfare, that is
%$\OPT = \sum_{i=1}^n\OPT_i\cdot\Pr[\ell = i]$.

We consider the following mechanism that takes as input a distribution $D$. The
mechanism is deterministic, so all probabilities are over the
distribution $D$. We note that the mechanism decides on a maximal number of items it is going to sell {\em without looking at the bids}.
Although it seems somewhat surprising it still achieves good approximation when the non-decreasing hazard rate condition holds.
\\

{\centering
\fbox{
\begin{minipage}{15cm}
\begin{tabbing}
\textbf{HazardGuess}($D$):\\
(x) \= \kill
 \> 1. Fix an arbitrary permutation $\pi$ on the bidders.\\
 \> 2. Solicit bids, and denote them $v_1,\ldots,v_n$ in non-increasing order.\\
 \> 3. Let $s^*$ be the smallest integer such that $s^* \geq \frac{\Pr[\ell \geq s^*]}{\Pr[\ell = s^*]}$. % \\ \> \;\;\;
 If $s^* > 3$ let $g = s^*$. Otherwise let $g = 1$.\ \footnote{Alternatively, we can pick $g=s^*$ always, but then we must pick a random permutation in step 1 of HazardGuess. We choose to present a deterministic mechanism.}\\
  \> 4. Consider only the highest $g$ bidders ordered according to $\pi$.\\
  \> \;\;\; When an item arrives sell it to the first of the remaining
such bidders and charge him $v_{g+1}$ (or 0\\
  \> \;\;\; if $g=n$).
  % If more than $g$ items arrive, do not assign any additional items.
%  \> 4. Assign each of the first $g$ items that come in to the highest $g$ bidders in the order in which they appear \\
%  \> \;\;\; in $\pi$, and charge them $v_{g+1}$. If more than $g$ items arrive, do not assign any additional items.
\end{tabbing}
\end{minipage}
}}
\newline

\begin{theorem}
\label{HazardGuessApprox} HazardGuess($D$) is truthful, %and online envy-free,
and achieves a $16\frac{7}{8}$-approximation
to social welfare in expectation over $D$, for any distribution $D$
such that the hazard rate $h_i(D)$ is non-decreasing.
\end{theorem}

Truthfulness %and envy-freeness are
is immediate: Every bidder with bid higher than $v_{g+1}$ faces a
single take-it-or-leave-it offer at the same price ($v_{g+1}$). The
offer and the order in which they receive the offer is independent
of their own bids.
%When $n$ items arrive, only bidders with bids at
%most the sale price are not offered items.
To prove the
approximation guarantee, we will need a series of lemmas.

The following lemmas, \ref{boundSlemma}, \ref{technicalLemma2} and
\ref{BestGuessApprox} will show that for any distribution with
non-decreasing hazard rate, $\max_i \OPT_i\cdot\Pr[\ell \geq i] \geq
\OPT/5$. To complete the proof, we will then prove that HazardGuess
achieves welfare at least $(8/27)\cdot\max_i \OPT_i\cdot\Pr[\ell
\geq i]$, and thus achieves a $16\frac{7}{8}$ approximation to
$\OPT$.

\begin{lemma}
\label{boundSlemma} Let $\alpha$ be the smallest value such that for
any set of bids, $\OPT/(\max_i \OPT_i\cdot\Pr[\ell \geq i]) \leq
\alpha$. Then for each integer $0 \leq s \leq n-1$ we have the
following bound on $\alpha$ in terms of $D$, which we denote
Bound$(s)$:
$$\alpha \leq \sum_{i=1}^s \frac{\Pr[\ell = i]}{\Pr[\ell \geq i]} + \frac{\sum_{i=s+1}^n\Pr[\ell = i]\cdot i}{(s+1)\cdot\Pr[\ell \geq s+1]}$$
\end{lemma}
\begin{proof}
Suppose $\alpha > \beta$. That is, there exists a set of bids such that for all $i$ we
have $\OPT_i\cdot\Pr[\ell \geq i] < \OPT/\beta$, or
equivalently:
\begin{equation}
\label{OPTiBound} \OPT_i < \frac{\OPT}{\beta\cdot\Pr[\ell \geq
i]}
\end{equation}
Recall that by definition, we have $\OPT = \sum_{i=1}^n\OPT_i\cdot\Pr[\ell = i]$.
Observe that for all $1 \leq i \leq n-1$:
$\OPT_{i+1} \leq \frac{i+1}{i}\OPT_i$
since $v_1,\ldots,v_n$ is a non-increasing sequence. By repeated application of this observation, we get the
following $n$ upper-bounds on $\OPT$ indexed by $0 \leq s \leq n-1$:
$$\OPT \leq \sum_{i=1}^s\OPT_i\cdot \Pr[\ell = i] + \OPT_{s+1}\cdot\left(\sum_{i=s+1}^n\frac{i}{s+1}\Pr[\ell = i] \right)$$
Applying inequality \ref{OPTiBound} and multiplying both sides by $\beta/\OPT$ we obtain:
$$\beta < \left(\sum_{i=1}^s\frac{\Pr[\ell = i]}{\Pr[\ell \geq i]} + \frac{\sum_{i=s+1}^n\Pr[\ell = i]\cdot i}{(s+1)\cdot \Pr[\ell \geq s+1]}\right).$$
If $\alpha$ is the optimal approximation factor, there is some input such that for every $\epsilon > 0$,$\max_i\OPT_i\cdot\Pr[\ell \geq i]$ achieves an $\alpha$ approximation but does not achieve a $\beta = \alpha-\epsilon$ approximation, and the above bound on $\beta$ holds. Since $\alpha = \beta + \epsilon$, letting $\epsilon$ tend to zero, we obtain the lemma.
\end{proof}

\begin{rem}
\label{remark:cannot-build-distribution-proof-mechanism}
We must now show that for every distribution $D$, there exists an $s$
such that Bound$(s)$ gives $\alpha \leq 5$. Note that the order
of quantifiers is important! It is not the case that there exists an
$s$ such that for every distribution, Bound$(s)$ gives $\alpha \leq
O(1)$.
%For example, for the uniform distribution, Bound$(m/2)$ gives
%$\alpha \leq O(1)$, but Bound$(1)$ gives $\alpha \leq \Omega(n)$.
%Conversely, for an exponential distribution, Bound$(m/2)$ gives
%$\alpha \leq \Omega(n)$, whereas Bound$(1)$ gives $\alpha \leq
%O(1)$.
\end{rem}

\begin{lemma}
\label{technicalLemma2} For any $s \geq 1$ and $h_i \in [1/s,1]$:
$\sum_{i=s+1}^n\left(i\cdot h_i\cdot \prod_{j=s+1}^{i-1}(1-h_j)\right) \leq 3s+1$.
\end{lemma}
\begin{proof}
We defer the proof of this technical lemma to Appendix~\ref{app:stoch-proofs}.
\end{proof}

\begin{lemma}
\label{BestGuessApprox} For any set of bids, and for any
distribution $D$ with non-decreasing hazard rate,\\
$\frac{OPT}{\max_i \OPT_i\cdot\Pr[\ell \geq i]} \leq 5$.
\end{lemma}
\begin{proof}
Given a distribution $D$, we wish to find the value of $s$ such that
Bound$(s)$ gives the sharpest bound on $\alpha$ (the approximation factor from lemma \ref{boundSlemma}).
%To find the correct value of $s$, we switch to the continuous setting. We let $F(x)$ be the CDF of $D$, and let $f(x) = F'(x)$. We then have a continuous version of Bound$(s)$, $\alpha \leq B(s)$ where:
%$$B(s) = \int_0^s\frac{f(\ell)}{1-F(\ell)}d\ell + \frac{\int_s^nf(\ell)\cdot\ell d\ell}{s\cdot(1-F(s))}$$
%We take the derivative with respect to $s$ to minimize $B(s)$:
%$$B'(s) = \int_s^nf(\ell)\cdot\ell d\ell\cdot \left(\frac{f(s)}{s\cdot (1-F(s)^2)} - \frac{1}{s^2(1-F(s))}\right)$$
%Setting $B'(s) = 0$, we find that $B(s)$ is minimized at $s = n$ or $s = (1-F(s))/f(s)$.
%
%We now move back into the discrete domain, and adapt the above bound:
We choose $s^* \leq n$ to be the smallest integer such that $s^*
\geq \Pr[\ell \geq s^*]/\Pr[\ell = s^*]$. If no such $s^*$ exists,
we choose $s^* = n$. We now show that Bound$(s^*)$ gives $\alpha \leq 5$.
We bound the two terms of Bound$(s^*)$ separately. Consider
the first term:
$$\sum_{i=1}^{s^*} \frac{\Pr[\ell = i]}{\Pr[\ell \geq i]} \leq (s^*-1)\cdot \frac{\Pr[\ell = s^*-1]}{\Pr[\ell \geq s^*-1]}+ \frac{\Pr[\ell = s^*]}{\Pr[\ell \geq s^*]} \leq 1 + \frac{\Pr[\ell = s^*]}{\Pr[\ell \geq s^*]} \leq 2$$
%\begin{eqnarray*}
%\sum_{i=1}^{s^*} \frac{\Pr[\ell = i]}{\Pr[\ell \geq i]} &\leq& (s^*-1)\cdot \frac{\Pr[\ell = s^*-1]}{\Pr[\ell \geq s^*-1]}+ \frac{\Pr[\ell = s^*]}{\Pr[\ell \geq s^*]} \\
%&\leq& 1 + \frac{\Pr[\ell = s^*]}{\Pr[\ell \geq s^*]} \;\;\; \leq \;\;\;  2
%%\\ &\leq& 2
%\end{eqnarray*}
since the hazard rate is non-decreasing and by definition of $s$. We
now consider the second term:
$\frac{\sum_{i=s^*+1}^n\Pr[\ell = i]\cdot i}{(s^*+1)\cdot\Pr[\ell \geq s^*+1]}$
Since $D$ has a non-decreasing hazard rate, we know that for all $i
\geq s^*$, $h_i \equiv \Pr[\ell = i]/\Pr[\ell \geq i] \geq 1/s^*$.
Therefore, we have:
\begin{eqnarray*}
\sum_{i=s^*+1}^n\Pr[\ell = i]\cdot i &=& \sum_{i=s^*+1}^n\frac{\Pr[\ell = i]}{\Pr[\ell \geq i]}\cdot\Pr[\ell \geq i]\cdot i \\
&=& \sum_{i=s^*+1}^n \left(i\cdot h_i\cdot \Pr[\ell \geq s^*+1]\cdot\prod_{j=s^*+1}^{i-1}(1-h_j)\right) \\
&\leq&\Pr[\ell \geq s^*+1](3s^*+1)
\end{eqnarray*}
where the inequality follows from Lemma~\ref{technicalLemma2}.
Therefore, finally we have for all $s^*$:
\begin{eqnarray*}
\frac{\sum_{i=s^*+1}^n\Pr[\ell = i]\cdot i}{(s^*+1)\cdot\Pr[\ell \geq s^*+1]} &\leq& \frac{\Pr[\ell \geq s^*+1](3s^*+1)}{(s^*+1)\cdot\Pr[\ell \geq s^*+1]} \leq 3
%\\ &\leq& 3
\end{eqnarray*}
Combining these two bounds, we finally get that Bound$(s^*)$ gives
$\alpha \leq 5$.
\end{proof}

Now we are ready to complete the proof of our theorem:

\begin{proof}[Proof of Theorem \ref{HazardGuessApprox}]
We show that HazardGuess achieves welfare at least $(8/27)\cdot(\max_i\OPT_i\cdot\Pr[\ell \geq i])$. Together with lemma \ref{BestGuessApprox}, this proves that HazardGuess achieves at least a $16\frac{7}{8}$ approximation to social welfare.

Let $s^*$ be the smallest integer such that $s^* \geq \frac{\Pr[\ell
\geq s^*]}{\Pr[\ell = s^*]}.$ Whenever $s^* > 3$, HazardGuess($D$)
achieves welfare at least $\OPT_{s^*}\cdot \Pr[\ell \geq s^* ]$.
When $s^* \leq 3$, HazardGuess$(D)$ achieves welfare at least
$\OPT_{s^*}/3$ (since it sells a single item to the highest bidder,
and $\OPT_1 \geq \OPT_3/3$). First consider the case in which $i >
s^* \geq 1$. In this case, we know $\Pr[\ell \geq i] \leq \Pr[\ell
\geq s^*]\cdot(1-\frac{1}{s^*})^{i-s^*}$, since the hazard rate
$h_i$ is non-decreasing, and $h_{s^*} \geq 1/s^*$. Therefore, we
have:
\begin{eqnarray*}
\OPT_i\cdot \Pr[\ell \geq i] &\leq& \frac{i}{s^*}\cdot \OPT_{s^*}\cdot \Pr[\ell \geq i] \\
&\leq& \frac{i}{s^*}\cdot \OPT_{s^*}\cdot \Pr[\ell \geq s^*]\cdot(1-\frac{1}{s^*})^{i-s^*} \\
&\leq& (\OPT_{s^*}\cdot  \Pr[\ell \geq s^*])\cdot \left(\frac{i}{s^*}\cdot \frac{1}{e^{i/s^* - 1}}\right) \\
&\leq& (\OPT_{s^*}\cdot  \Pr[\ell \geq s^*])
\end{eqnarray*}
Therefore, in this case, HazardGuess$(D)$ achieves welfare at least $\OPT_i\cdot \Pr[\ell \geq i]/3$.
Now consider the case in which $1 \leq i < s^*$: By definition of
$s^*$:
$\frac{\Pr[\ell \geq s^*-1]}{\Pr[\ell = s^*-1]} > s^*-1.$
Alternatively, we may write the hazard rate at $s^*-1$: $h_{s^*-1} <
1/(s^*-1).$ Since the hazard rate is non-decreasing, we have that
for all $i \leq s^*-1$, $h_i < 1/(s^*-1)$. Therefore we have:
\begin{eqnarray*}
\Pr[\ell \geq s^*]
\;\;\; = \;\;\;  \prod_{i=1}^{s^*-1}(1-h_i)
\;\;\; > \;\;\; \prod_{i=1}^{s^*-1}(1-\frac{1}{s^*-1})
\;\;\; = \;\;\; \left(\frac{s^*-2}{s^*-1} \right)^{s^*-1}
\end{eqnarray*}

If $s^* \geq 4$, then this gives $\Pr[\ell \geq s^*] \geq 8/27$.
Therefore:
\begin{eqnarray*}
\OPT_{s^*}\cdot\Pr[\ell \geq s^*]
\;\;\; \geq \;\;\; \OPT_i \cdot \Pr[\ell \geq 4]
\;\;\; \geq \;\;\; \frac{8}{27}\OPT_i
\end{eqnarray*}
which is a bound on the performance of HazardGuess$(D)$, since $s^* > 3$.
Finally we consider the special case of $s^* \in\{2,3\}$. If $s^* =
2$, then $i \in \{1,2\}$
achieves welfare $\OPT_i/2\cdot \Pr[\ell \geq i]$ since HazardGuess sells one item. Similarly, if $s^* = 3$ HazardGuess achieves
welfare at least $\OPT_i/3\cdot\Pr[\ell \geq i]$. This concludes the proof.
\end{proof}

We note that our analysis is worst-case, and that this mechanism can be shown to achieve a better constant approximation for specific distributions of interest. For example:
\begin{theorem}
\label{uniformConstant}
HazardGuess($D$) achieves a $\frac{3}{5}$-approximation to social welfare in expectation over $D$ when $D$ is the uniform distribution over $\{1,\ldots,n\}$. Moreover, there are values for which HazardGuess$(D)$ cannot get better than a $\frac{3}{4}$-approximation when $D$ is the uniform distribution.
\end{theorem}
The proof is deferred to the appendix.

%Intuition for the choice of $s^*$ in the TruthfulGuess algorithm can be found in the literature of revenue-maximizing auctions:
%
%\begin{rem}
%Readers that are familiar with Myerson's work on optimal auctions
%\cite{M81} might have noticed that the supply level chosen by
%TruthfulGuess is actually exactly the supply level where the
%\emph{virtual valuation} (that corresponds to the distribution $D$)
%equals zero (see \cite{Krishna02} for a survey), i.e.,
%$s-\frac{\Pr[\ell = s^*]}{\Pr[\ell \geq s^*]} \cong 0 $. That is,
%the optimal supply level is determined exactly in the way that the
%optimal reservation price is determined in revenue-maximizing
%auctions. This is no coincidence due to the following intuition:
%assume that all the bidder values were equal (say, 1). Our above analysis
%actually ignores the welfare gained by realizations where the supply
%is smaller than $s^*$. Therefore, the social welfare that we
%actually consider is $s^*$ when $\ell \geq s^*$ and zero otherwise.
%Now, consider the problem of setting the optimal reserve price
%when selling an item to a single bidder. For a reservation price
%$s^*$ we get a \emph{revenue} of $s^*$ if his value is greater than
%$s^*$ and zero revenue otherwise. It follows that the optimal supply level $s^*$ is calculated similarly to the optimal reservation price. Note that Myerson's result also requires a monotone hazard-rate condition. Our proof shows that this selection is a good choice even with arbitrary bidder values.
%
%\end{rem}

\subsection{The Necessity of the Monotone Hazard Rate Condition}
We next show that the monotone hazard rate condition is necessary: for arbitrary distributions \emph{no} deterministic mechanism can achieve constant approximation to social welfare.

\begin{theorem}
\label{thm:non-mhr-lb}
No deterministic truthful mechanism can achieve an $o(\sqrt{\log n/\log\log n})$ approximation to social welfare when faced with arbitrary stochastic supply (without the non-decreasing hazard rate condition).
\end{theorem}
We prove this theorem in Appendix~\ref{app:stoch-proofs}. The proof proceeds in two stages. First, we consider a class of truthful mechanisms that fix -independently of the bids- an ordering $\pi$ on the bidders, and a supply $g$. Such a mechanism  sells the first $g$ items that arrive at the $(g+1)$-st highest price to the $g$ highest bidders , ordered according to $\pi$. We note that HazardGuess is such a mechanism, and all such mechanisms satisfy a notion of envy-freeness which we define in the next section.
We show that such mechanisms cannot achieve an $o(\log n/\log\log n)$ approximation to social welfare when faced with arbitrary stochastic supply. We then complete the proof by showing that we can restrict our attention to such mechanisms almost without loss of generality: for \emph{any} deterministic mechanism, there exists a mechanism that chooses its supply independently of the bids that loses only a quadratic factor in its approximation to social welfare.

\section{Envy-Free Mechanisms}
\label{section:envy-free}

All our mechanisms satisfy a notion of fairness which is our
adaptation of envy-freeness to the online setting. An offline
mechanism is envy-free if no agent prefers another agent's
allocation and payment to his own (see, for example, \cite{GH03,GH05}). In the case of unit demand
bidders and identical goods this means that there is a price $p$
such that any winner pays the same price $p$ and has value at least
$p$, and any loser has value at most $p$. This is clearly not
possible to achieve for online supply, except by trivial mechanisms
(for example, the mechanism that only sells a single item to the highest
bidder at the second highest price). Informally, in an online envy-free
mechanism, the only source of envy is a shortage of supply, not
price discrimination on the part of the mechanism.

\begin{definition}
A deterministic mechanism is {\em online-envy-free} if it is
envy-free (in the offline sense) when the supply is enough to
satisfy the demand of all of the bidders (that is, when $l=n$).
A randomized mechanism is {\em online-envy free} if it is a distribution
over deterministic online-envy-free mechanisms.
\end{definition}

Note that this definition ensures that all sold items are sold
for the same price, even when the supply is smaller than $n$.
%In
%Section \ref{section:envy-free}, we characterize the set of
%online-envy-free mechanisms, and strengthen some of our lower
%bounds.
% We note that the mechanisms that we present are all online envy-free. We discuss this definition more in Section \ref{section:envy-free}; for now, the following definition will suffice.
%\begin{definition}
% A deterministic mechanism is \emph{online envy-free} if every bidder who is allocated an item pays the same price $p$, and if when $n$ items arrive, all bidders with bids greater than $p$ are allocated an item.
% \end{definition}
%% Moshe: added:
%A randomized mechanism should satisfy the condition for \emph{every} realization of its randomness.
%\begin{rem}
%In an envy-free mechanism, no bidder prefers the allocation and price of any other bidder. This is not possible to achieve by non-trivial mechanisms when we have online supply. An online-envy-free mechanism guarantees that the only source of envy is a shortage of supply, and not any price discrimination on the part of the mechanism.
%\end{rem}
Also note that both our mechanisms RandomGuess and HazardGuess are online-envy-free.\\

%as they fit into a common framework.
%We will show that the framework EnvyFreeScheme below fully characterizes truthful online-envy-free mechanisms.\\

%{\centering
%\fbox{
%\begin{minipage}{13cm}
%\begin{tabbing}
%\textbf{EnvyFreeScheme}:\\
%(x) \= \kill
% \> 1. Solicit bids, and denote them $v_1,\ldots,v_n$ in non-increasing order. \\ \> \;\;\; For simplicity we assume bids are distinct.\\
% \> 2. Fix a permutation $\pi$ on the bidders such that for all $i$, $\pi(i)$ is independent of $v_i$.\\
% \> 3. Choose a price $p$ (possibly as function of the bids).\\ \> \;\;\; If $p \in (v_g,v_{g+1}]$ then $p$ must be independent of  $v_1,\ldots,v_g$.\footnote{with the convention that $v_0=\infty$ and $v_{n+1}=0$.} \\
% \> 4. Consider only the highest $g$ bidders ordered according to $\pi$.\\
% \> \;\;\;  When an item arrives sell it to the first of the remaining such bidders and charge him $p$.
% % \> 4. Sell each of the first $g$ items that come in to the highest $g$ bidders in the order in which they appear in $\pi$,\\ \> \;\;\; and charge them $p$. If more than $g$ items arrive, do not sell any additional items.
%\end{tabbing}
%\end{minipage}
%}}
%\newline\\

%We first claim that EnvyFreeScheme characterizes truthful and online-envy-free mechanisms.

%\begin{proposition}
%\label{prop:truthful-envy}
%% For any procedure to choose $g$ in a way that is independent of the bids,
%An online supply mechanism is truthful and online-envy-free if and only if it is a parameterizations of EnvyFreeScheme.
%\end{proposition}

In Theorem~\ref{generalLowerBound} we showed that no truthful randomized mechanism can achieve an $o(\log\log n)$ approximation to social welfare when faced with adversarial supply. Here, we present an improved lower bound for truthful online-envy free mechanisms.

%1111 Next, we show that in the stochastic setting, we may (almost) without loss of generality restrict our attention to a subset of online-envy free mechanisms. Finally, we use this characterization to show that without the monotone hazard rate condition, no deterministic mechanism can achieve better than an $\Omega(\sqrt{\log n/\log\log n})$ approximation to social welfare in the stochastic setting.
%For space constraints, we defer some proofs in this section to appendix Section \ref{app:envy-free}.

\begin{proposition}
\label{envyfreeThm}
No truthful online-envy-free mechanism (even randomized) can achieve an $o(\log n/\log\log n)$ approximation to social welfare when faced with adversarial supply.
\end{proposition}
We defer the proof to appendix Section \ref{app:envy-free}. Note that proposition \ref{envyfreeThm} is nearly tight, since RandomGuess achieves a $\log n$ approximation factor.

\section{Valuations with complementarities: Knapsack Valuations}
\label{sec:knapsack}
% Since our mechanisms can handle bidders with arbitrary non-increasing marginal utility functions,
So far we have discussed bidders that are interested in a single item out of a set of identical items.
It is natural to consider the case of bidders with increasing-marginal utility valuations,
corresponding to \emph{complements} valuations. In the extreme case, we get \emph{knapsack valuations}.
% Such valuations model bidders that would like to run an entire campaign and such campaign only achieves its goals if it is successful in getting some minimal number of impressions (items).

% While non-increasing marginal utilities correspond to the economic notion of \emph{substitutes}, increasing marginal utilities correspond to \emph{complements}. In the extreme case, we get \emph{knapsack valuations}.

We say that a bidder $i$ has a \emph{knapsack valuation} if he has a
value $c_i$ and a desired quantity $k_i$: For all $k < k_i$, $v_i(k)
= 0$, and for all $k \geq k_i$, $v_i(k) = c_i$. That is, bidder $i$
desires at least $k_i$ units of the good, is not satisfied
with fewer, and has no value for more than $k_i$ units.

 Knapsack valuations can be seen as modeling advertising campaigns: a buyer wishes to build brand
name recognition through banner-advertisements, and so has little
value for a small number of advertisements; A campaign is worth
$c_i$ to the advertiser, but additional advertising saturation has
little added benefit.

Unfortunately, the online nature of the problem makes knapsack
valuations difficult to handle for any algorithm, even without
truthfulness (and computational) constraints. Here, we present an algorithm in the
stochastic setting, and show that its (poor) competitive ratio is
optimal over the class of all (not necessarily truthful) algorithms. Without loss of generality, we can
assume that $D$ has finite support over $[1,m]$ for $m=\sum_{i=1}^n k_i$.

Our lower bound for Knapsack valuations shows that with online
supply, no algorithm can guarantee a better approximation ratio than
the cumulative hazard rate. This welfare guarantee is quite poor.
For the uniform distribution, this gives $\alpha$=$\Theta(\log m)$.
For the binomial distribution, $\alpha = \Theta(m)$. We also present a matching upper bound showing that our lower bound is tight. Both proofs are in Appendix \ref{app:knapsack-proofs}.

\begin{proposition}
\label{knapsackLowerBound}
No algorithm can have %guarantee
better than a $\sum_{i=1}^m h_i$ approximation to optimal social welfare.
\end{proposition}

%\newline\newline
%\textbf{KnapsackGuess}($D$):
%\begin{enumerate}
%  \item Solicit bids. For each bidder $i$, create a knapsack instance with one item corresponding to each bidder $i$, with size $k_i$ and value $c_i$. For each $s \in [1,m]$ let $\OPT_s$ be the value of the optimal solution to this knapsack instance when the knapsack has size $s$.
%  \item Let $s^* = \arg\max_s \Pr[\ell \geq s]\cdot \OPT_s$.
%  \item Assign items as they arrive to bidders corresponding to the optimal solution for a knapsack of size $s^*$ in an arbitrary order, until each bidder $i$ in the solution has received his demand, $k_i$ items.
%\end{enumerate}
%

\begin{proposition}
\label{knapsackUpperBound}
For any distribution $D$ with (arbitrary) hazard rate $h_i$ there exists an algorithm that achieves at least a $\sum_{i=1}^mh_i$ approximation to optimal social welfare.
\end{proposition}

% \subsection{Multiple Item Types}
%We now consider a model with multiple item types, $T_1,\ldots,T_m$. The supply of each item type $T_i$ is drawn from an independent distribution $D_i$. After the supply is determined for each type of item, the items are presented to the mechanism in an arbitrary order. Bidders continue to be unit demand; Each bidder has a valuation $v_i$ and an interest set, $S_i \subset \{T_1,\ldots,T_m\}$. Bidder $i$ gets value $v_i$ if he receives any item in the set $S_i$.

\bibliographystyle{plain}
\bibliography{onlinesupply}

%\section{Appendix}
\appendix
\section{Proofs}

\subsection{Proof: Lower bound for deterministic mechanism with adversarial supply}
\label{app:det-adversarial}
In this section we prove Theorem~\ref{thm:det-advers-lb}.
\begin{theorem}
No deterministic truthful mechanism can achieve better than an $n$ approximation to social welfare.
\end{theorem}
The theorem will follow from three simple lemmas.
\begin{lemma}
For every truthful mechanism and for any realization of items, the price $p_b$ that bidder $b$ is charged upon winning (any) item is independent of his bid.
\end{lemma}
\begin{proof}
This is a standard fact characterizing truthful auctions; If there is some realization of items for which bidder $b$ has two distinct bids which result in bidder $b$ winning an item, but at a different price, then in the case in which his valuation is equal to the bid that yields an item at the higher price, he will report falsely that his valuation is equal to the bid that yields an item at the lower price.
\end{proof}
\begin{lemma}
For every truthful mechanism and for any realization of items, if bidder $b$ wins an item, which item bidder $b$ wins is independent of his bid whenever $p_b < v_b$.
\end{lemma}
\begin{proof}
Suppose for some realization of items, and for some fixed set of bids of the other bidders, bidder $b$ can change his bid to $v_b$ or $v_b'$, and win one of two items, item $i$ or item $j$, and that if he bids his true valuation $v_b$, he wins item $j > i$. Now consider a realization in which only $i$ items arrive; If bidder $b$ bids $v_b$, he wins no item and receives utility $0$. If he bids $v_b'$, he wins item $i$ at his (bid independent) price $p_b$, and achieves higher utility $v_b-p_b$. Therefore, the mechanism is not truthful.
\end{proof}
\begin{lemma}
\label{approximationLemma}
For any deterministic mechanism that achieves an $n$-approximation to social welfare, every bidder has a bid such that they are allocated the first item.
\end{lemma}
\begin{proof}
Any bidder $b$ can set his bid to more than $n$ times the second highest bidder. If the mechanism does not allocate the first item to $b$, then if there are no further items, the mechanism has not achieved an $n$-approximation to social welfare.
\end{proof}
\begin{proof}[Proof of Theorem]
By Lemma~\ref{approximationLemma}, any bidder can win the first item
with an appropriately high bid. But by Lemma~\ref{orderLemma}, any
bidder such that $p_b < v_b$ who has a bid for which he can win the
first item cannot win any other item with any bid. Therefore, for
any  set of bidders $b_i$ such that for all $b_i$, $p_{b_i} \ne
v_{b_i}$, then any deterministic truthful mechanism that
achieves an $n$-approximation can only sell the first item. If all
bidders have value $1 \leq v_{b_i} \leq 1+\epsilon$, this achieves
no better than an $n$-approximation when all items arrive. It
remains to demonstrate such a set of bidders: Consider an arbitrary
set of $n+1$ distinct values between $1$ and $1+\epsilon$. For each
bidder, choose a value from this set independently at random. Since
each bidders price $p_{b_i}$ is independent of his bid, by
Lemma~\ref{priceLemma}, the probability that $v_{b_i} = p_{b_i}$ is at
most $1/(n+1)$, and by the union bound, the probability that
\textit{any} bidders bid equals its price threshold is at most
$n/(n+1) \leq 1$. Therefore, there exists a set of bids sampled from
this set with the desired property, which completes the proof.
\end{proof}

\subsection{Proofs of Lemmas from Section \ref{AdversarialSection}}
\label{app:adver-proofs}

\textbf{Lemma \ref{OPTLemma}:}
Consider a set of $n$ valuations drawn from $V$ and let $\OPT_k$ denote the sum of the $k$ highest valuations from the set. Then:
$$\E[\OPT_k] \geq H_{k+1}-1.$$
where $H_{k+1}$ denotes the $k+1$st harmonic number. In particular, $\E[\OPT_k] > (\log k)/2$.
\newline
\begin{proof}
Let $F(y)$ denote the cumulative distribution function of $V$. We note that $F(y)$ is a step function taking values $F(y) = (n-1/y)/(n-1)$ for all $y$ of the form $y = 1/2^i$ for $i \in \{0, 1, \ldots, \log n-1\}$. We consider the inverse CDF function $F^{-1}(x):[0,1]\rightarrow\{1, 1/2, 1/4, \ldots, 2/n\}$. It is simple to verify the following pointwise lower bound on $F^{-1}(x)$:
$$F^{-1}(x) \geq \frac{1}{n-x(n-1)}$$
which follows from inverting the discrete CDF. We denote the quantity in this bound $A(x) = 1/(n-x(n-1))$, and observe that $A(x)$ is convex in the range $[0,1]$.

Let $v_{i,n}$ denote the $i$'th largest value out of $n$ draws from $V$, and let $X_{i,n}$ denote the $i$'th largest value out of $n$ draws from the uniform distribution over $[0,1]$. We consider the following method of drawing a value $v$ from $V$: we draw $x$ uniformly from $[0,1]$ and let $v = F^{-1}(x)$. Since $F^{-1}$ is monotone, the $i$'th largest draw from the uniform distribution corresponds to the $i$'th largest draw from $V$: $v_i = F^{-1}(x_i)$.

Recall the expected value of the $i$'th largest of $n$ draws from the uniform distribution over $[0,1]$: $E[X_{i,n}] = 1-i/(n+1)$. This standard fact follows from a simple symmetry argument. We are now ready to complete the proof of the lemma:
\begin{eqnarray*}
\E[\OPT_k] &=& \sum_{i=1}^k\E[v_{i,n}] \\
&=& \sum_{i=1}^k\E[F^{-1}(X_{i,n})] \\
&\geq& \sum_{i=1}^k E[A(X_{i,n})] \\
&\geq& \sum_{i=1}^k A(E[X_{i,n}]) \\
&=& \sum_{i=1}^k \frac{1}{1+i(n-1)/(n+1)} \\
&\geq& \sum_{i=1}^k \frac{1}{1+i} \\
&=& H(k+1)-1
\end{eqnarray*}
where the second inequality is an application of Jensen's inequality, which follows since $A(x)$ is convex.
\end{proof}

\textbf{Lemma \ref{notTooBigLemma}:}
For $c_i \in [0,\log n-1]$:
$$\sum_{i=b_k+1}^n\frac{(c_i+1)}{\exp(\frac{\sum_{j=b_k+1}^{i-1}2^{c_j}}{n-1})} < 2.5\cdot n$$
\newline
\begin{proof}
Let $f(c_{b_k+1},\ldots,c_n)) \equiv \sum_{i=b_k+1}^n\frac{(c_i+1)}{\exp(\frac{\sum_{j=b_k+1}^{i-1}2^{c_j}}{n-1})}$. We consider the partial derivative at the $i$'th offer price:
\begin{eqnarray*}
\frac{\partial}{\partial c_i}f(c_{b_k+1},\ldots,c_n) &=& \frac{1}{e^{\sum_{j=b_k+1}^{i-1}2^{c_j}/(n-1)}} - \left(\frac{2^{c_i}\ln 2}{(n-1)e^{2^{c_i}/(n-1)}} \right)\cdot\sum_{j=i+1}^n\frac{c_j+1}{\exp(\sum_{\ell=b_k+1\ \ell \ne i}^{j-1}2^{c_\ell}/(n-1))} \\
&\leq& 1 - \frac{\ln 2}{n-1}\cdot\left(\sum_{j=i+1}^n\frac{c_j+1}{\exp(\sum_{\ell=b_k+1\ \ell \ne i}^{j-1}2^{c_\ell}/(n-1))} \right)
\end{eqnarray*}
But this is negative unless
$$R_i \equiv \sum_{j=i+1}^n\frac{c_j+1}{\exp(\sum_{\ell=b_k+1\ \ell \ne i}^{j-1}2^{c_\ell}/(n-1))} \leq \frac{n-1}{\ln 2}$$

Fixing any maximal assignment to the $c_i$ variables, let $i'$ be the largest index for which the above condition on $R_{i'}$ fails to hold. We know that for all $i \leq i'$, $c_i = 0$, since the partial derivative at $i$ is negative, and so if we could reduce $c_i$ further this would contradict the fact that we selected a maximal assignment. Therefore, we have:
\begin{eqnarray*}
f(c_{b_k+1},\ldots,c_n) &=& \sum_{i=b_k+1}^{i'}\frac{c_i+1}{\exp(\frac{\sum_{j=b_k+1}^{i-1}2^{c_j}}{n-1})} + \sum_{i=i'+1}^n\frac{c_i+1}{\exp(\frac{\sum_{j=b_k+1}^{i-1}2^{c_j}}{n-1})} \\
&\leq& i' + \frac{1}{e^{2^{c_i}/(n-1)}}\cdot R_i \\
&\leq& n + \frac{n-1}{\ln 2} \\
&<& 2.5 n
\end{eqnarray*}

\end{proof}

\begin{proposition}[Proposition \ref{RandomGuessFactor}] RandomGuess is truthful, online-envy-free, and
achieves a $\log n$ approximation to social welfare.
\end{proposition}
\begin{proof}
Truthfulness and envy-freeness are immediate: every winning bidder faces a single take-it-or-leave-it offer independent of their bid, in an order independent of their bid. All items are sold at the same price, $v_{g+1}$. When $n$ items arrive, all bidders with valuations higher than the offer price have been allocated items. We now prove the approximation guarantee.

Suppose that $I$ items arrive, and $\OPT_I = \sum_{i=1}^Iv_i$, the sum of the $I$ highest bids. With probability $1/\log n$, $I < g \leq 2I$, and with probability $1/\log n$, $I/2 < g \leq I$. In the first case, RandomGuess allocates the $I$ items to at least half of the top $g$ bidders in random order, and so achieves welfare in expectation at least $\OPT_g/2 \geq OPT_I/2$. In the second case, RandomGuess allocates at least half of the $I$ items to all of the top $g$ bidders, and achieves welfare $\OPT_g = \sum_{i=1}^gv_i$. Since $g > I/2$, $\OPT_g > \OPT_I/2$ because $\{v_i\}$ is a non-increasing sequence. Our mechanism therefore achieves in expectation welfare at least $(1/\log n)(\OPT_I/2 + \OPT_I/2) = \OPT_I/\log n$.
\end{proof}

\subsection{Proofs from Section \ref{StochasticSection}}
\label{app:stoch-proofs}

\textbf{Lemma \ref{technicalLemma2}}: For any $s \geq 1$ and $h_i \in [1/s,1]$:
$$\sum_{i=s+1}^n\left(i\cdot h_i\cdot \prod_{j=s+1}^{i-1}(1-h_j)\right) \leq 3s+1$$
\newline
\begin{proof}
Let $f(h_{s+1},\ldots,h_n) \equiv \sum_{i=s+1}^n(i\cdot h_i\cdot
\prod_{j=s+1}^{i-1}(1-h_j))$ and consider the partial derivative at
$h_k$:
\begin{eqnarray*}
\frac{\partial}{\partial h_k}f(h_{s+1},\ldots,h_n)&=& k\cdot\prod_{j=s+1}^{k-1}(1-h_j)-\sum_{i=k+1}^n\left(i\cdot h_i\cdot \prod_{j = s+1, j \ne k}^{i-1}(1-h_j)\right) \\
 &\leq& k\cdot(1-\frac{1}{s})^{k-s-1} - \sum_{i=k+1}^n\left(i\cdot h_i\cdot \prod_{j = s+1, j \ne k}^{i-1}(1-h_j)\right)
\end{eqnarray*}
where the inequality follows from $h_i \geq 1/s$ for all $i$. But
this is negative unless
$$R_k \equiv \sum_{i=k+1}^n\left(i\cdot h_i\cdot \prod_{j = s+1, j \ne k}^{i-1}(1-h_j)\right) \leq k\cdot(1-\frac{1}{s})^{k-s-1}$$
Fix some assignment to the $h_i$ that maximizes
$f(h_{s+1},\ldots,h_n)$ and let $k'$ be the first index at which the
above condition holds. Then for all $i < k'$, $h_i = 1/s$, since
otherwise this would contradict the fact that the assignment
maximizes $f$. Therefore, we have:
\begin{eqnarray*}
 \sum_{i=s+1}^n\left(i\cdot h_i\cdot \prod_{j=s+1}^{i-1}(1-h_j)\right) &=& \sum_{i=s+1}^{k'-1}\left(i\cdot h_i\cdot \prod_{j=s+1}^{i-1}(1-h_j)\right) + \sum_{i=k'}^n\left(i\cdot h_i\cdot \prod_{j=s+1}^{i-1}(1-h_j)\right) \\ %liad changed h_i--> h_j
 &\leq& \sum_{i=s+1}^{k'-1}\left(\frac{i}{s}(1-\frac{1}{s})^{i-s-1}\right) +  \sum_{i=k'}^n\left(i\cdot h_i\cdot \prod_{j=s+1}^{i-1}(1-h_j)\right) \\
 &=& \sum_{i=s+1}^{k'-1}\left(\frac{i}{s}(1-\frac{1}{s})^{i-s-1}\right)+ k'\cdot h_{k'} \cdot \prod_{j=s+1}^{k'-1}(1-h_j) + (1-h_{k'})\cdot R_{k'} \\
 &\leq& \sum_{i=s+1}^{k'-1}\left(\frac{i}{s}(1-\frac{1}{s})^{i-s-1}\right) + h_{k'}(\cdot k'\cdot (1-\frac{1}{s})^{k'-s-1}) +  (1-h_{k'})(k'\cdot (1-\frac{1}{s})^{k'-s-1}) \\
 &=& \frac{1}{s}\sum_{i=s+1}^{k'-1}\left(i(1-\frac{1}{s})^{i-s-1}\right) + k'\cdot(1-\frac{1}{s})^{k'-s-1} \\
 &\leq& \frac{1}{s}\sum_{i=s+1}^{\infty}\left(i(1-\frac{1}{s})^{i-s-1}\right) + (s+1) \\
 &=& 3s+1
\end{eqnarray*}
where the second inequality follows from the fact that for all $i$,
$h_i \geq 1/s$, the third inequality follows from the fact that $k
\geq s+1$ and so $k'\cdot(1-\frac{1}{s})^{k'-s-1}$ is decreasing in
$k'$, and the last equality follows from the identity
$\sum_{i=k}^\infty i\cdot r^{i-k} = (k + r - kr)/(r-1)^2$.
\end{proof}

\textbf{Theorem \ref{uniformConstant}}:HazardGuess($D$) achieves a $\frac{3}{5}$-approximation to social welfare in expectation over $D$ when $D$ is the uniform distribution over $\{1,\ldots,n\}$. Moreover, there are values for which HazardGuess$(D)$ cannot get better than a $\frac{3}{4}$-approximation when $D$ is the uniform distribution.\newline
\begin{proof}
Consider the case that there are $n$ agents and the supply is chosen uniformly at random from $\{1,n\}$ (we note that if the range starts from a number larger than $1$ the problem becomes easier and the algorithm achieves better approximation.)
We analyze the approximation achieved by picking the supply $k = n/2$ and selling at most $k$ items,\footnote{For simplicity we assume that $n$ is even. Essentially the same argument will work for the case that $n$ is odd.} in a random order over the top $k$ values. We prove that the algorithm achieves at least $60\%$ of the optimum.

Assume the values are sorted $v_1\ge v_2\ge\ldots\ge v_n$. Define $OPT_l=\sum_{i=1}^l v_i$. The expected welfare of the optimal algorithms is $OPT=1/n \cdot \sum_{l=1}^n OPT_l$. Splitting the sum to two parts we get the following.

$$ OPT = % \frac{1}{n} \cdot \sum_{l=1}^n OPT_l =
\frac{1}{n} \cdot \sum_{l=1}^{\frac{n}{2}} OPT_{l} +
\frac{1}{n} \cdot \sum_{l=\frac{n}{2}+1}^{n} OPT_{l}\le
\frac{OPT_{\frac{n}{2}}}{2} +
\frac{1}{n} \cdot \sum_{l=\frac{n}{2}+1}^{n} \frac{l}{n/2} OPT_{\frac{n}{2}} =
OPT_{\frac{n}{2}} \left(\frac{1}{2} + \frac{2}{n^2} \sum_{l=\frac{n}{2}+1}^{n} l   \right) =
$$
$$
OPT_{\frac{n}{2}} \left(\frac{1}{2} + \frac{2}{n^2} \left(\frac{n(n+1)}{2} - \frac{\frac{n}{2}(\frac{n}{2}+1)}{2}\right)\right) =
OPT_{\frac{n}{2}} \left(\frac{5}{4} + \frac{1}{2n}\right)
$$

Our algorithm achieves expected welfare of
$$ ALG = \frac{1}{n} \cdot \sum_{l=1}^{\frac{n}{2}} \frac{l}{n/2} OPT_{\frac{n}{2}} +
\frac{1}{n} \cdot \sum_{l=\frac{n}{2}+1}^{n} OPT_{\frac{n}{2}}=
OPT_{\frac{n}{2}} \left(\frac{2}{n^2} \sum_{l=1}^{\frac{n}{2}} l   + \frac{1}{2}\right)=
$$
$$
OPT_{\frac{n}{2}} \left(\frac{2}{n^2} \frac{\frac{n}{2}(\frac{n}{2}+1)}{2}   + \frac{1}{2}\right)=
OPT_{\frac{n}{2}} \left(\frac{3}{4} +\frac{1}{2n}\right) \ge
\frac{OPT}{\frac{5}{4} + \frac{1}{2n}}\cdot \left(\frac{3}{4} +\frac{1}{2n}\right)\ge
\frac{3}{5} OPT$$

Finally we observe that this algorithm gets at most $75\%$ of the optimum. Consider the input with one value of $1$ and all the rest of the values are 0. The optimal algorithm will always get welfare of $1$. Our algorithm will get the $1$ with probability
$$\sum_{l=1}^{n/2} \frac{1}{n} \cdot \frac{l}{n/2} + \frac{1}{2} = \frac{n+1}{4n}  +  \frac{1}{2}< \alpha$$ for any constant $\alpha>3/4$ when $n$ is large enough.
\end{proof}

\textbf{Theorem \ref{thm:non-mhr-lb}}:
No deterministic truthful mechanism can achieve an $o(\sqrt{\log n/\log\log n})$ approximation to social welfare when faced with arbitrary stochastic supply (without the non-decreasing hazard rate condition).
\newline
The theorem follows directly from two lemmas.

\begin{definition}
A \emph{bid-independent supply mechanism} chooses an ordering on the bidders $\pi$ and a supply $g$ independently of the bids. It then sells items as they arrive to the $g$ highest bidders, ordered according to $\pi$, at the $g+1$st highest price.
\end{definition}
Note that all mechanisms presented in this paper are bid-independent supply mechanisms.
\begin{lemma}
\label{necessaryHazardRate}
No deterministic bid-independent supply mechanism can achieve an $o(\log n/\log\log n)$ approximation to social welfare when faced with arbitrary stochastic supply (without the non-decreasing hazard rate condition).
\end{lemma}
\begin{proof}
We give a distribution with a decreasing hazard rate such that no
mechanism that determines a maximum supply $g$ independent of the
bids $v_i$ can achieve an $o(\log n/ \log \log n)$ approximation to social
welfare.

We define $D$ such that $\Pr[\ell = i] = 1/(i + i^2)$. Note that $\Pr[\ell \geq i] =
1/i$, and the hazard rate at $i$ is decreasing: $h_i(D) = 1/(1+i)$.
Consider the welfare achieved by a bid-independent mechanism that chooses supply $g$. If at least $g$
items arrive, it achieves welfare exactly $\OPT_g$. Otherwise, if $j
< g$ items arrive, it achieves expected welfare at most $(j/g)\OPT_g$.
Therefore, the welfare it achieves is at most:
\begin{eqnarray*}
\OPT_g\cdot\Pr[\ell \geq g] + \frac{1}{g}\cdot\sum_{j=1}^{g-1}j\cdot\Pr[\ell = j] &=& \OPT_g\cdot(\frac{1}{g} + \frac{H_g-1}{g}) \\
&=& \Theta\left(\OPT_g\cdot(\frac{\log g}{g})\right)
\end{eqnarray*}
We consider two possible sets of bidder values: In the Single Bidder
case, we have $v_1 = 1$ and $v_j = 0$ for all $j > 1$. In the All
Bidder case, we have $v_j = 1$ for all $j$. Note that in the Single
Bidder case, we have $\OPT = 1$ and $\OPT_i = 1$ for all $i$. In the
All Bidder case we have $\OPT = H_{n+1}-1 = \Theta(\log n)$ and
$\OPT_i = i$. Therefore, in the Single Bidder case, a mechanism that
achieved an $o(\log n/\log\log n)$ approximation to social welfare
would have $(\log g)/g = \omega(\log\log n/\log n)$, and in the All
Bidder case would have $\log g = \omega(\log\log n)$. There is no $g
\in [1,n]$ that satisfies both of these equations simultaneously.
Since $g$ is chosen independently of the bids, the two cases are indistinguishable, and any such mechanism much achieve an
approximation ratio no better than $\Omega(\log n/\log\log n)$ in at
least one of them.
\end{proof}

\begin{lemma}
For any distribution $D$ and any deterministic truthful mechanism $M$ that achieves an $\alpha$ approximation to social welfare over $D$, there is a truthful deterministic online-envy-free bid-independent supply mechanism $M'$
that achieves an $\alpha^2$ approximation to social welfare.
\end{lemma}
\begin{proof}
Let $g_{\max}$ be the maximum number of items $M$ sells when full supply is realized, where the maximum is taken over all possible bid profiles. Let $M'$ be the mechanism that always sells the first $g_{\max}$ items to the $g_{\max}$ highest bidders in some predetermined order at the $g_{\max}+1$st highest price, and sells no further items. Note that $M'$ is online-envy-free and
has bid-independent sell sequence.
% chooses its supply independently of the bids.
First observe that $\OPT_{g_{\max}} \geq \OPT/\alpha$. This follows because by definition, $M$ can never achieve welfare beyond $\OPT_{g_{\max}}$, but by assumption, $M$ achieves an $\alpha$ approximation to the optimal social welfare. Next, observe that $\Pr_D[\ell \geq g_{\max}] \geq 1/\alpha.$ To see this, consider some bid profile which causes $M$ to produce a supply $g_{\max}$. Let $b_i$ be the bidder who receives item $g_{\max}$, and consider raising his valuation $v_i$ until it constitutes all but a negligible fraction of the total possible social welfare. By lemmas \ref{priceLemma} and \ref{orderLemma}, raising $b_i$'s bid does not affect either the supply offered by the mechanism, or the order in which $b_i$ receives an item: that is, it continues to be the case that $b_i$ receives an item if and only if at least $g_{\max}$ items arrive. However, since $b_i$ now constitutes an arbitrarily large fraction of the total social welfare, and $M$ is an $\alpha$-approximation mechanism, it must be that $Pr[\ell \geq g_{\max}] \geq 1/\alpha$.

Finally, we observe that our mechanism achieves welfare at least $\OPT_{g_{\max}}\cdot\Pr[\ell \geq g_{\max}] \geq \OPT/\alpha^2$, which completes the proof.
\end{proof}

\subsection{Proofs from Section \ref{section:envy-free}}
\label{app:envy-free}
%
%\begin{proposition}[Proposition \ref{prop:truthful-envy}]
%% For any procedure to choose $g$ in a way that is independent of the bids,
%An online supply mechanism is truthful and online envy-free if and only if it is a parameterizations of EnvyFreeScheme.
%\end{proposition}
%\begin{proof}
%We first show that the scheme is truthful and envy free.
%Truthfulness is immediate: Bidders who receive offers face take-it-or-leave-it prices that are independent of their bids. Which item is offered to them is also independent of their bids. Envy-freeness is also clear; All items are sold for the same price $p$, and when $n$ items arrive, every bidder $i$ with $v_i > p$ is offered the item.
%
%Now consider an arbitrary mechanism that is truthful and envy free, we show that it is a parametrization of EnvyFreeScheme. We solicit bids, and let $\pi$ be the order in which items are allocated when $n$ items arrive (if the mechanism does not allocate all items, we may fill in the remainder of permutation $\pi$ arbitrarily). Since we begin with a truthful mechanism, by Lemma \ref{orderLemma}, this ordering must have the property that $\pi(i)$ is independent of $v(i)$ for all $i$. Similarly, let $p$ be the single price at which the mechanism sells all items (since it is envy-free). By lemma \ref{priceLemma}, this price must be independent of any winning bidder's bid.
%\end{proof}

\begin{proposition}[Proposition \ref{envyfreeThm}]
No truthful online-envy-free mechanism can achieve an $o(\log
n/\log\log n)$ approximation to social welfare when faced with
adversarial supply.
\end{proposition}
\begin{proof}
For an envy-free mechanism, we may assume that all offered prices
$c_1,\ldots,c_n$ are equal: for all $i$, $c_i = c$. We apply
inequality \ref{lowerBoundConstraint} to obtain constraints for the
case in which $n$ items arrive, and the case in which $1$ item
arrives. When $n$ items arrive, we have for all $i$ $\Pr[N_{i-1} <
n] = 1$, and obtain the constraint:
\begin{equation}
\label{envyFreeConstraint1} n\cdot c \geq \frac{(n-1)\log
n}{2\alpha} - n
\end{equation}
When a single item arrives, we have $\Pr[N_{i-1} < 1] =
((n-2^{c+1})/(n-1))^{i-1}$, since each bidder independently accepts
the offer price $1/2^c$ with probability $(2^{c+1}-1)/(n-1)$.
Also, $\OPT_1\ge 1/2$.
We obtain the constraint:
\begin{equation} \label{envyFreeConstraint2}
(c+1)\cdot \sum_{i=1}^n\left(\frac{n-2^{c+1}}{n-1}\right)^{i-1} \geq
\frac{n-1}{2\alpha}
\end{equation}
Setting $\alpha = o(\log n/\log\log n)$, we see that constraint
\ref{envyFreeConstraint1} requires $c = \omega(\log\log n)$. It is
simple to verify that the left hand side of constraint
\ref{envyFreeConstraint2} is decreasing in $c$ in the range
$[\log\log n,\log(n)-1]$, and that setting $c = \omega(\log\log n)$
fails to satisfy \ref{envyFreeConstraint2}, which proves the claim.
\end{proof}

\subsection{Proofs from Section \ref{sec:knapsack}}
\label{app:knapsack-proofs}

We begin by presenting a lower bound for Knapsack utilities.
\begin{proposition}[Proposition \ref{knapsackLowerBound}]
No algorithm can guarantee better than a $\sum_{i=1}^m h_i$ approximation to optimal social welfare.
\end{proposition}
\begin{proof}
Consider any arbitrary distribution $D$ and scale it so that it has
positive support on $[m+1,2m]$. Alternately, imagine it has positive
support on $[1,m]$, and that $m$ items are guaranteed to arrive; the
distribution is on how many additional items will arrive. We
construct a set of $n=m$ bidders $1,\ldots,m$. Bidder $i$ has $k_i =
m+i$ and $c_i = 1/\Pr[\ell \geq i]$. By construction, at most one
bidder can have his demand satisfied by any knapsack size. Since
bidder values are non-decreasing, we have
$$\OPT = \sum_{i=1}^m
c_i\cdot\Pr[\ell = i] = \sum_{i=1}^m\frac{\Pr[\ell = i]}{\Pr[\ell
\geq i]} = \sum_{i=1}^mh_i$$ However, since at most one bidder can
be satisfied by any knapsack size, no algorithm can do better than
picking some bidder $i$ and assigning all items that arrive to
bidder $i$. Such an algorithm achieves welfare $c_i$ in the case
that $k_i$ items arrive. By construction, this yields expected
welfare $(1/\Pr[\ell \geq i])\cdot\Pr[\ell \geq i] = 1$, which
completes the proof.
\end{proof}

{\centering
\fbox{
\begin{minipage}{12cm}
\begin{tabbing}
\textbf{KnapsackGuess}($D$):\\
(x) \= \kill
 \> 1. Solicit bids. For each bidder $i$, create a knapsack instance with one item corresponding to each bidder $i$,\\ \> \;\;\; with size $k_i$ and value $c_i$. For each $s \in [1,m]$ let $\OPT_s$ be the value of the optimal solution to \\
  \> \;\;\; this knapsack instance when the knapsack has size $s$.\\\\
 \> 2. Let $s^* = \arg\max_s \Pr[\ell \geq s]\cdot \OPT_s$.\\\\
 \> 3. Assign items as they arrive to bidders corresponding to the optimal solution for a knapsack of size $s^*$ \\
 \> \;\;\;  in an arbitrary order, until each bidder $i$ in the solution has received his demand, $k_i$ items.
\end{tabbing}
\end{minipage}
}}
\newline

\begin{rem}
Rather than solving the knapsack problem exactly to find $\OPT_s$, we can use the greedy-by-density algorithm to find a 2-approximation.
\footnote{The greedy-by-density algorithm first discard all items of size larger than the knapsack size and then picks the best of the following two allocations: the greedy-by-density allocation that picks requests in decreasing ratio of value to size until the next element does not fit, and the allocation that gives all the items to the request of highest value.}
It is simple to see that the greedy knapsack algorithm can only ever output at most $2n$ distinct solutions, regardless of knapsack size. Therefore, at the cost of a factor of $2$, our algorithm only has to consider $2n$ solutions, each of which can be computed in polynomial time.
\end{rem}

\begin{proposition}[Proposition \ref{knapsackUpperBound}]
For any distribution $D$ with (arbitrary) hazard rate $h_i$ KnapsackGuess($D$) achieves at least a $\sum_{i=1}^mh_i$ approximation to optimal social welfare.
\end{proposition}
\begin{proof}
KnapsackGuess($D$) achieves welfare $\OPT_{s^*}$ whenever $s^*$ items arrive, which occurs with probability $\Pr[\ell \geq s^*]$. Therefore, KnapsackGuess achieves welfare at least $\OPT_{s^*}\cdot \Pr[\ell \geq s^*] \geq \OPT_{s'}\cdot\Pr[\ell \geq s']$ for all $s'$. Let $\OPT$ denote the expected optimal welfare when the number of items to be sold is drawn from $D$. If KnapsackGuess achieves no better than an $\alpha$ approximation to social welfare, then for all $s' \in [1,m]$: $\OPT_{s'}\cdot\Pr[\ell \geq s'] \leq \OPT/\alpha$, or equivalently:
$$\OPT_{s'} \leq \frac{\OPT}{\alpha\Pr[\ell \geq s']}.$$
By definition:
$$\OPT = \sum_{i=1}^m\OPT_i\cdot\Pr[\ell = i].$$
Using our above bound on $\OPT_i$:
$$\OPT \leq \sum_{i=1}^m \OPT\cdot\frac{\Pr[\ell  = i]}{\alpha\Pr[\ell \geq i]}.$$ Therefore:
$$\alpha \leq \sum_{i=1}^m\frac{\Pr[\ell  = i]}{\Pr[\ell \geq i]} = \sum_{i=1}^mh_i$$
which completes the proof.
\end{proof}

\end{document}